\documentclass[13pt]{article}
\usepackage{amsfonts}
\usepackage{mathrsfs}
\usepackage{amssymb,amsmath} 
\usepackage{color}
\usepackage{graphicx}
\usepackage{float}
 \allowdisplaybreaks

\catcode`!=11
\let\!int\int \def\int{\displaystyle\!int}
\let\!lim\lim \def\lim{\displaystyle\!lim}
\let\!sum\sum \def\sum{\displaystyle\!sum}
\let\!sup\sup \def\sup{\displaystyle\!sup}
\let\!inf\inf \def\inf{\displaystyle\!inf}
\let\!cap\cap \def\cap{\displaystyle\!cap}
\let\!max\max \def\max{\displaystyle\!max}
\let\!min\min \def\min{\displaystyle\!min}
\let\!frac\frac \def\frac{\displaystyle\!frac}
\catcode`!=12

\let\oldsection\section
\renewcommand\section{\setcounter{equation}{0}\oldsection}

\allowdisplaybreaks

\newtheorem{defn}{Definition}[section]
\newtheorem{thm}{Theorem}[section]
\newtheorem{pro}{Proposition}[section]
\newtheorem{lem}{Lemma}[section]
\newtheorem{re}{Remark}[section]
\newenvironment{proof}{{\noindent \it \bf Proof :}}{{\hfill$\Box$}\\}
\title{Rogue peakon, well-posedness, ill-posedness and blow-up phenomenon for an integrable Camassa-Holm  type equation}
\author{
{\bf\large Mingxuan Zhu}\\
\small School of Mathematical Sciences\\
\small Qufu Normal University, Qufu, 273100, P. R. China\\
\small E-mail: mxzhu@qfnu.edu.cn\\
{\bf\large Zhenteng Zeng}\\
\small School of Mathematical and Statistical Sciences \\
\small University of Texas Rio Grande Valley, Edinburg, TX 78539, USA \\
\small E-mail: zhenteng.zeng01@utrgv.edu \\
{\bf\large Zaihong Jiang }\\
\small Department of Mathematics \\
\small Zhejiang Normal University, Jinhua, 321004, P. R. China \\
\small E-mail: jzhong@zjnu.cn \\
{\bf\large Baoqiang Xia}\\
\small School of Mathematics and Statistics \\
\small Jiangsu Normal University, Xuzhou, 221116, P. R. China \\
\small E-mail: xiabaoqiang@126.com \\
{\bf\large Zhijun Qiao\thanks{Corresponding author} }\\
\small School of Mathematical and Statistical Sciences \\
\small University of Texas Rio Grande Valley, Edinburg, TX 78539, USA \\
\small E-mail: zhijun.qiao@utrgv.edu}
\date{}

\begin{document}

\maketitle

\begin{abstract} In this paper, we study an integrable Camassa-Holm (CH) type equation with quadratic nonlinearity.  The CH type equation is shown integrable through a Lax pair, and particularly the equation is found to possess a new kind of peaked soliton (peakon) solution - called {\sf rogue peakon}, that is given in a rational form with some logarithmic function, but not a regular traveling wave. We also provide multi-rogue peakon solutions. Furthermore, we discuss the local well-posedness of the solution in the Besov space $B_{p,r}^{s}$ with $1\leq p,r\leq\infty$, $s>\max \left\{1+1/p,3/2\right\}$ or $B_{2,1}^{3/2}$, and then prove the ill-posedness of the solution in $B_{2,\infty}^{3/2}$. Moreover, we establish the global existence and blow-up phenomenon of the solution, which is, if $m_0(x)=u_0-u_{0xx}\geq(\not\equiv) 0$, then the corresponding solution  exists globally, meanwhile, if $m_0(x)\leq(\not\equiv) 0$, then the corresponding solution blows up in a finite time.


\end{abstract}
\section{Introduction and rogue peakons}
In this paper, we consider the following Camassa-Holm (CH) type equation that was proposed recently in \cite{XQZ2013, Zeng2023}
\begin{eqnarray} \label{1.0}
m_t+mu+(mu)_x=0,\quad m=u-u_{xx}.
\end{eqnarray}
The above CH type equation (\ref{1.0}) is actually a reduced case of the following
integrable two-component peakon system \cite{XQZ2013}
\begin{align}\label{2component}
    \begin{cases}
    m_{t}=F+F_x- \frac{1}{2}m(uv-u_x v_x +uv_x-u_x v),\\
     n_{t}=-G+G_x+ \frac{1}{2}u(uv-u_x v_x +uv_x-u_x v),\\
     m=u-u_{xx},\\
     n=v-v_{xx},
\end{cases}
\end{align}
where $F$ and $G$ are two arbitrary functions of $u$,\,$v$ and their derivatives satisfying $mG=nF.$
Let $v=0$, $G=0$, and $F=-mu$ in (\ref{2component}), then Eq. (\ref{2component}) is exactly reduced to the scalar model (\ref{1.0}).
Through a straightforward computation, it is not hard for us to see the scalar CH type equation (\ref{1.0}) possesses the following Lax pair $\psi_x=U\psi,\: \psi_{t}=V\psi,$
\begin{equation}\label{Laxpair}
U=\frac{1}{2} \begin{pmatrix}
  -1 & \lambda m\\
  0 & 1
\end{pmatrix},\quad  V=-\frac{1}{2}\begin{pmatrix}
    \lambda^{-2} & -\lambda^{-1}(u-u_x)+\lambda mu\\
    0& -\lambda^{-2}
\end{pmatrix},
\end{equation}
where $\psi$ is the 2-dimensional eigenvector, $\lambda$ refers to the eigenvalue parameter, and $U, V$ satisfy the following compatibility condition -- zero curvature representation $U_{t}-V_{x}+[U,V]=0.$ Therefore, the CH type equation (\ref{1.0}) is completely integrable in the sense of Lax pair.\\

{\bf Rogue Peakon}

The most attractive property of the CH type equation (\ref{1.0}) is that it has the following so-called {\sf Rogue Peakon} (see the formula given in Eq. (\ref{roguepeakon}) below and {\bf Remark 1}).

Let $u=p(t)e^{-|x-q(t)|}$. One may  compute
\begin{eqnarray*}
 u_x=-pe^{-|x-q|}sign(x-q), \nonumber\\
m=u-u_{xx}=2p\delta(x-q), \nonumber\\
m_t=2p'\delta(x-q)-2pq'\delta'(x-q)\nonumber\\
m_x=2p\delta'(x-q),
\end{eqnarray*}
where $\delta$ the Dirac delta and $\delta'$ is its derivative.
Substituting the above formulas to \eqref{1.0}, we get
\begin{eqnarray*}
(2p'+2p^2e^{-|x-q|}-2p^2e^{-|x-q|}sign(x-q))\delta(x-q)-2p(q'-pe^{-|x-q|})\delta'(x-q)=0.
\end{eqnarray*}
Evaluating this on a test function $\varphi$ and integratig by parts, we have
\begin{eqnarray*}
\langle(2p'+2p^2e^{-|x-q|}-4p^2e^{-|x-q|}sign(x-q))\delta(x-q),\varphi\rangle+\langle2p(q'-pe^{-|x-q|})\delta(x-q),\varphi'\rangle=0.
\end{eqnarray*}
The above equation  holds for any $\varphi$. Thus, $p(t)$ and $q(t)$ satisfy
\begin{eqnarray*}
p'=-p^2,\quad q'=p.
\end{eqnarray*}
Then, we arrive at
\begin{eqnarray}
 p(t)&=&\frac{1}{t-A},\quad q(t) \ = \ \ln|t-A|+B, \nonumber\\
 u(x,t) &=& \frac{1}{t-A}e^{-|x-\ln|t-A|-B|}, \label{roguepeakon}
\end{eqnarray}
where $A$, $B$ are two arbitrary constants. The figure of one- rogue peakon (\ref{roguepeakon}) is plotted below.
\\
\begin{figure}[H]
\begin{center}
 \includegraphics[height=4.6cm,width=6cm]{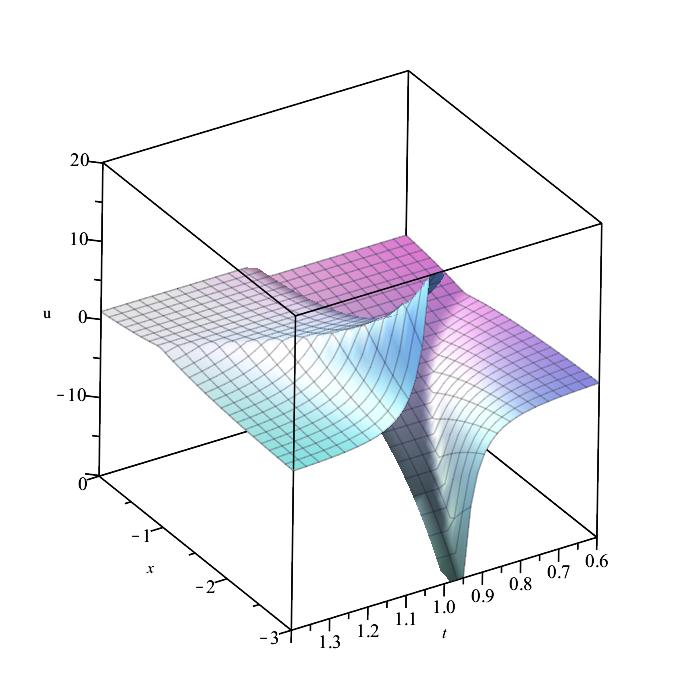}
\end{center}
\caption{The single rogue peakon solution (\ref{roguepeakon})}
with $A=1, B=0$ and plotting for $x\in (-3,0), \ t\in (0.6,1.4)$.
Apparently, the solution $u(x,t)$ is continuous everywhere for $x$ and $t$ except $t=1$ because it has a finite jump $2e^{-x}$ at $t=1$ (see {\bf Remark 1} below).
\label{fig:OneRP}
\end{figure}

$  $

{\bf Remark 1:} Apparently, the single peakon solution given in Eq. (\ref{roguepeakon}) is not a traveling solitary wave. The amplitude of the peakon is the rational function $p(t)= \frac{1}{t-A}$ that has a sole singular point at $t=A$. The solution $u(x,t)$ is continuous everywhere for both variables $x$ and $t$ except $t=A$ because it has a finite jump at $t=A$ (see the left and right limits computed below):
\begin{align*}
u(A^{-})&=\lim_{t\to A^{-}} u =\lim_{t\to A^{-}} \frac{e^{-|x-\ln|t-A|-B|}}{t-A},\\
    & =-e^{-x+B},\\
u(A^{+})&=\lim_{t\to A^{+}} u =\lim_{t\to A^{+}} \frac{e^{-|x-\ln|t-A|-B|}}{t-A},\\
    & =e^{-x+B}.
\end{align*}
This tells us that the peakon amplitude suddenly changes from negative to positive nearby $t=A$, that is, right after the crest, it is followed up immediately by the trough. Since $p(t)$ is a rational function and the peakon solution is in a non-traveling wave shape, let us call this wave solution (\ref{roguepeakon}) as {\sf Rogue Peakon}. This is the significant difference from all existing peakon solutions in literature.\\

{\bf Multi-Rogue Peakons}

 \noindent Let us now seek for multi-rogue peakon solutions in the form of
\begin{equation}\label{multi-rogue}
  u(x,t)=\sum_{k=1}^{N}p_{k}(t)e^{-|x-q_{k}(t)|}.
\end{equation}
Substituting Eq. (\ref{multi-rogue}) to the CH type equation (\ref{1.0}), we obtain the following $2N$ ODE system in $t$:
\begin{eqnarray} \label{2NODE}
\left\{
  \begin{array}{ll}
    \dot{p}_{k}(t)=-p_{k}(t)\sum_{j=1}^{N}p_{j}(t)e^{-|q_{k}(t)-q_{j}(t)|}, \\
    \dot{q}_{k}(t)=\sum_{j=1}^{N}p_{j}(t)e^{-|q_{k}(t)-q_{j}(t)|}.
  \end{array}
\right.
\end{eqnarray}
Solving the above system (\ref{2NODE}) yields solutions
\begin{eqnarray*} 
\left\{
  \begin{array}{ll}
p_{k}(t)= \frac{1}{Nt-A_k}\\
q_{k}(t)=\ln |Nt-A_k| +B_k,\quad k=1,2,...,N,
 \end{array}\right.
\end{eqnarray*}
where $A_k$ and $B_k$ are arbitrary constants and $N$ refers to number of peaks. Therefore, we get the $N$-rogue peakon solutions (\ref{multi-rogue}) solved with
\begin{eqnarray*}
  u(x,t)&=& \sum_{k=1}^{N}\frac{e^{-|x-\ln |Nt-A_k| -B_k|}}{Nt-A_k}.
\end{eqnarray*}

As we see the $N$ rogue peakon solution $u(x,t)$ solved above, it can be understood as a linear superposition of non-traveling peaked waves in $x$ with coefficients of $t$-related rational form for an nonlinear intgerable PDE. This property is usually true for a linear differential equation, instead of a nonlinear PDE. But here in our paper, we find the rogue peakon has the linear superposition for our integrable nonlinear CH type equation (\ref{1.0}), which is the significant difference from other multi-soliton interactions and nonlinear PDEs. For two- rogue peakons ($N=2$, see Figure \ref{fig:TwoRP})
\begin{eqnarray} \label{multi-rogue2}
 u(x,t) = \frac{e^{-|x-\ln |2t-A_1| -B_1|}}{2t-A_1} +\frac{e^{-|x-\ln |2t-A_2| -B_2|}}{2t-A_2},
   \end{eqnarray}
and three- rogue peakons ($N=3$, see Figure \ref{fig:ThreeRP})
\begin{eqnarray} \label{multi-rogue3}
 u(x,t) = \frac{e^{-|x-\ln |3t-A_1| -B_1|}}{3t-A_1} +\frac{e^{-|x-\ln |3t-A_2| -B_2|}}{3t-A_2}+\frac{e^{-|x-\ln |3t-A_3| -B_3|}}{3t-A_3},
   \end{eqnarray}
 we have the graphs plotted below.

\begin{figure}[H]
\begin{minipage}[t]{0.5\linewidth}
\begin{center}
\includegraphics[height=4.0cm,width=5.5cm]{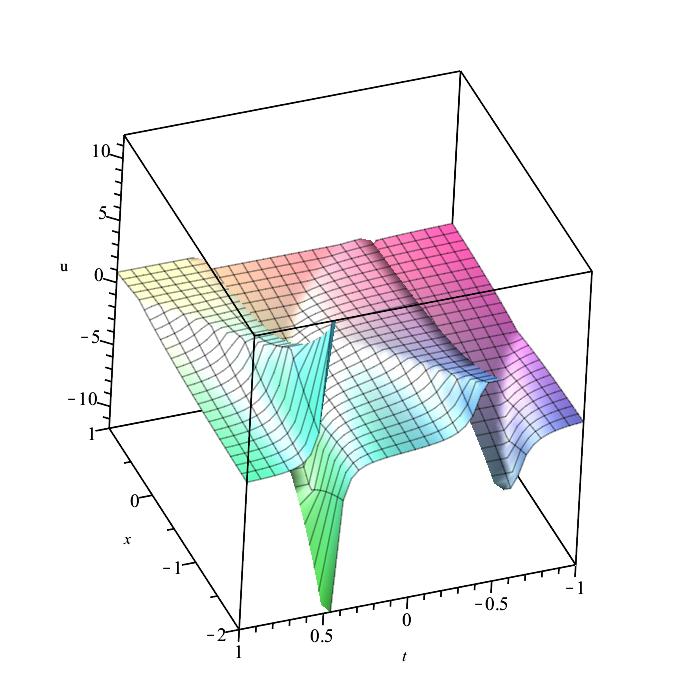}
 \end{center}
\caption{The two-rogue peakon solution (\ref{multi-rogue2})} with $A_1=-1, A_2=1, B_1=-0.5, B_2=0.5$ and plotting for $x\in (-2,1), \ t\in (-1,1)$.
Apparently, the solution $u(x,t)$ is continuous everywhere for $x$ and $t$ except $t=\pm 0.5$ because it has a finite jump $2e^{-x+ 0.5}$ at $t=0.5$ and $2e^{-x-0.5}$ at $t=-0.5$, respectively.
\label{fig:TwoRP}
\end{minipage}
\hspace{1.9ex}
\begin{minipage}[t]{0.5\linewidth}
\begin{center}
\includegraphics[height=4.0cm,width=5.5cm]{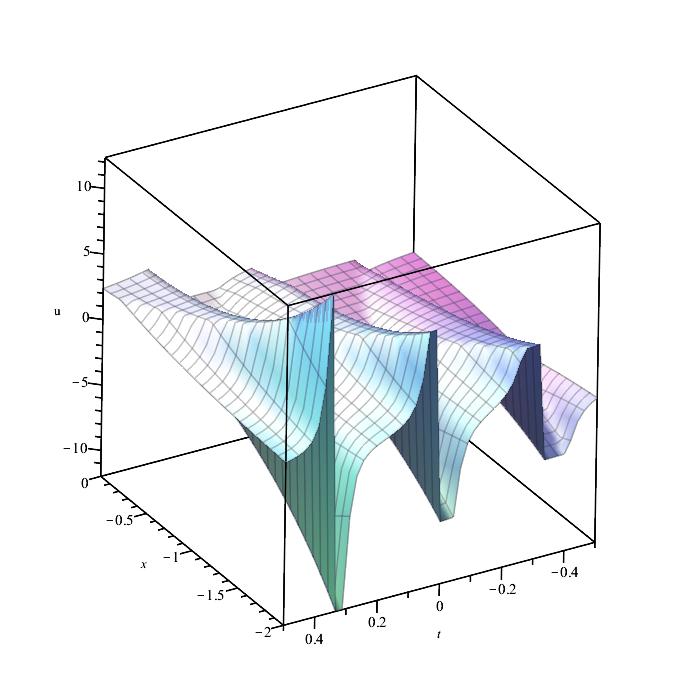}
 \end{center}
\caption{ The three-rogue peakon solution (\ref{multi-rogue3})} with $A_1=-1, A_2=0, A_3=1, B_1=-0.5, B_2=0, B_3=0.5$ and plotting for $x\in (-2,0), \ t\in (-0.5,0.5)$.
Apparently, the solution $u(x,t)$ is continuous everywhere for $x$ and $t$ except $t=0, \pm 0.33$ because it has a finite jump $2e^{-x}$ at $t=0$, $2e^{-x+0.33}$ at $t=0.33$ and $2e^{-x-0.33}$ at $t=-0.33$, respectively.
\label{fig:ThreeRP}
\end{minipage}
\end{figure}

 We already noticed that this is a linear superposition phenomenon with $t$-fractional coefficients for nonlinear peakons crossing each peak $x=\ln |Nt-A_k|+B_k$, but the shape of the rogue peakon between any two consecutive peaks
 is still taking on a dynamical interaction, and none of them are traveling wave type. See this phenomena in the following three figures plotted for $2-$rogue peakons.

 \begin{figure}[H]
\begin{minipage}[t]{0.33\linewidth}
\includegraphics[height=3.6cm,width=5cm]{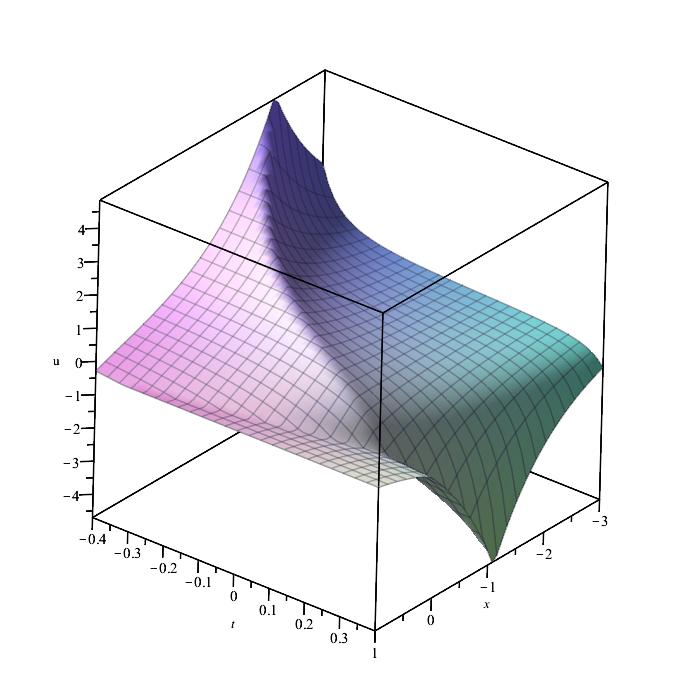}
\caption{ \small The two- rogue peakon solution (\ref{multi-rogue2}) with $A_1=-1, A_2=1, B_1=-0.5, B_2=0.5$ and plotting for $x\in (-3,1), \ t\in (-0.4,0.4)$.
Apparently, the solution $u(x,t)$ is continuous for $x\in (-3,1), \ t\in (-0.4,0.4)$ and two peakons (actually peakon and anti-pekaon) get interacted from positive to negative in this region.} 
\end{minipage}
\hspace{1.9ex}
\begin{minipage}[t]{0.33\linewidth}
\includegraphics[height=3.60cm,width=5.0cm]{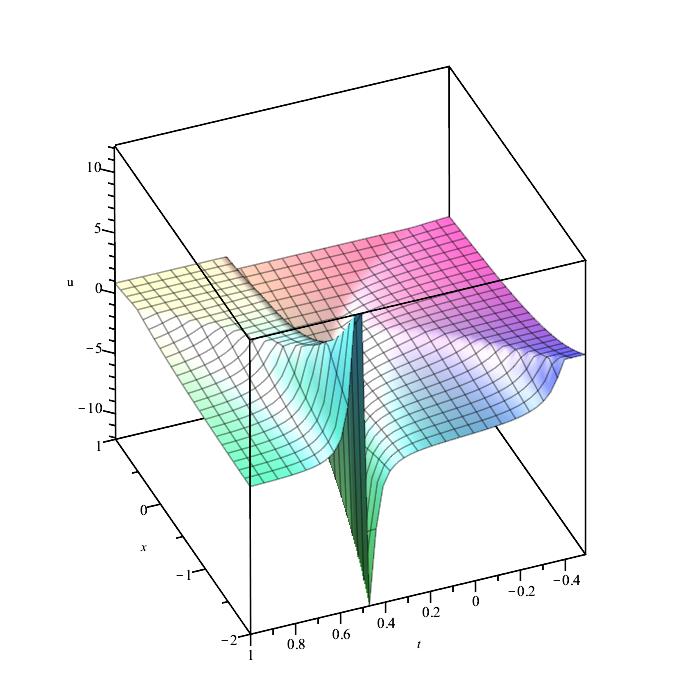}
\caption{ \small The two- rogue peakon solution (\ref{multi-rogue2}) with $A_1=-1, A_2=1, B_1=-0.5, B_2=0.5$ and plotting for $x\in (-2,1), \ t\in (-0.49,1)$.
Apparently, the solution $u(x,t)$ is continuous everywhere for $x$ and $t$ except $t=0.5$ because it has a finite jump $2e^{-x+ 0.5}$ at $t=0.5$.} But at the moment of going to and departing from 
$t=0.5$, two rogue peakons still get interacted. 
\end{minipage}
\hspace{1.9ex}
\begin{minipage}[t]{0.33\linewidth}
\includegraphics[height=3.60cm,width=5.0cm]{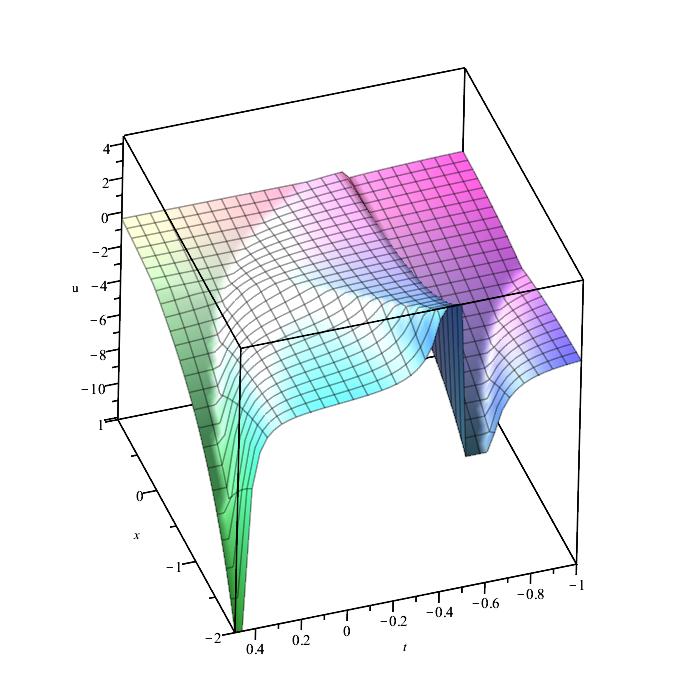}
\caption{ \small The two- rogue peakon solution (\ref{multi-rogue2}) with $A_1=-1, A_2=1, B_1=-0.5, B_2=0.5$ and plotting for $x\in (-2,1), \ t\in (-1,0.49)$.
Apparently, the solution $u(x,t)$ is continuous everywhere for $x$ and $t$ except $t=-0.5$ because it has a finite jump $2e^{-x- 0.5}$ at $t=-0.5$. But at the moment of going to and departing from 
$t=-0.5$, two rogue peakons still get interacted.}
\end{minipage}
\end{figure}

{\bf Background}

The system \eqref{1.0} is highly related to the celebrated Camassa-Holm(CH) equation
\begin{eqnarray*}
m_t+2mu_x+m_x u=0,\quad m=u-u_{xx},
\end{eqnarray*}
which was introduced in \cite{RDDH} and was also obtained by Fokas and Fuchssteiner \cite{FF} as a 
nonlinear partial differential equation with bi-Hamiltonian structures and inherited symmetries. It was shown 
completely integrable with the remarkable properties
of possessing 
Lax pair \cite{RDDH}, great features on the circle \cite{CM}, the integrable hierarchy extension with finite-dimensional constrained flows and algebro-geometric solutions on a symplectic submanifold connected to integrable negative order flows \cite{Qiao2003}, inverse scattering transform \cite{CGI} and multi-soliton solutions via Darboux transformation \cite{XQZ2016}, infinitely many conservation laws \cite{Le1} and geometric background of the CH equation \cite{Re}, and explicit solutions including classical soliton, cuspon and peakon solutions \cite{RDDH, Le2, QZ2006}.

Local well-posedness for the initial data in $H^s$ with $s>3/2$ was proved in \cite{CE1,LO}, while Danchin have extended the well-posedness to the Besov space $B_{p,r}^{s}$ ($s>\max \left\{1+1/p,3/2\right\}$, $1\leq p,r\leq\infty$) or $B_{2,1}^{3/2}$ in \cite{Dan2001,Dan2003}. There have been extensive studies on Blow-up phenomenon of solutions, cf. \cite{Bra,C,CE1,CE2,JNZ,LO,Mc} and references therein. For global existence, McKean \cite{Mc} (See also \cite{JNZ} for a simple proof) established a necessary and sufficient condition on the initial value $u_0$, which depends on the shape of $m_0=u_0-u_{0xx}$. More precisely, it was proved that the Camassa-Holm equation breaks if and only if some portion of the positive part of $m_0$ lies to the left of some portion of its negative part. Global conservative solution and global dissipative solution for the Camassa-Holm equation have been shown in \cite{BC,BC1}. As for the asymptotic properties,  persistence properties and unique continuation of solutions of the Camassa-Holm equation were given in \cite{HMPZ}, the long time behavior for the support of momentum density of the Camassa-Holm equation was discussed in \cite{JZZ}, and the orbital stability of the peakons was proved in \cite{CS}. 

$ $

{\bf Remark 2:} An 
observation with
 the substitution $m=u-u_{xx}=\widetilde{m} e^{-x}$, our equation  \eqref{1.0} is sent 
 to the following form
\begin{eqnarray*}
\widetilde{m}_t+\widetilde{m}_x u+\widetilde{m}u_x=0,  \quad 
m=u-u_{xx}.
\end{eqnarray*}
This form exactly falls in the Holm-Staley $b$-family equation \cite{HS} when $b=1$
\begin{eqnarray}\label{bfamily}
m_t+m_xu+bmu_x=0,\quad m=u-u_{xx}.
\end{eqnarray}
As shown in \cite{VN2009}, the Holm-Staley $b$-family equation (\ref{bfamily}) is integrable {\sf only} when $b=2$ and $b=3$, namely,
$b=2$ is the case for the integrable Camassa-Holm equation \cite{RDDH} while $b=3$ is the case for the integrable Degasperis-Procesi equation \cite{ADMP}. Although our equation  \eqref{1.0} can be transformed to the $b$-family with $b=1$, we already show that our equation \eqref{1.0} is completely integrable in the sense of Lax pair (see Eq. (\ref{Laxpair}) listed above).      
More mathematical properties for the Holm-Staley $b$-family equation were investigated in \cite{EY,GLT,Zhou}.

$ $

The aim of this paper is to study the Cauchy problem for system \eqref{1.0},
\begin{eqnarray} \label{1.1}
\left\{
  \begin{array}{ll}
&m_t+mu+(mu)_x=0,\quad m=u-u_{xx},\quad t\in\mathbb{R}^+,\quad x\in\mathbb{R}\\
&u(x,0)=u_0(x),
\end{array}
\right.
\end{eqnarray}
and establish the local and global well-posedness, ill-posedness and blow-up phenomenon of solutions. Although our system is similar to the Holm-Staley $b$-family equation, we obtained some completely different results.

The rest of this paper is organized as follows. In Section \ref{LW}, we discuss the local well-posedness of the solutions in some Besov spaces. Section \ref{IP} aims at proving the ill-posedness in $B_{2,\infty}^{3/2}$. Sufficient conditions on the initial data to guarantee global existence and blow-up phenomenon are established in Section \ref{SGB}. In the Appendix, we recall some conclusions on the properties of Littlewood-Paley decomposition, the nonhomogeneous Besov spaces and the theory
of the transport equation.

\section{Local well-posedness}\label{LW}
The main goal of this section is to establish the local well-posedness of the Cauchy problem for \eqref{1.0} in the nonhomogeneous Besov spaces, for which we introduce the following spaces.
\begin{defn}
For $T>0$, $s\in\mathbb{R}$ and $1\leq p\leq +\infty$, we set
\begin{eqnarray*}
&&E_{p,r}^s(T)=:C([0, T];B_{p,r}^{s})\cap C^1([0, T];B_{p,r}^{s-1}),\quad if~~ r<\infty,\\
&&E_{p,\infty}^s(T)=:L^\infty([0, T];B_{p,\infty}^{s})\cap Lip([0, T];B_{p,\infty}^{s-1}),
\end{eqnarray*}
and \begin{eqnarray*} E_{p,r}^s=:\bigcap_{T>0}E_{p,r}^s(T).\end{eqnarray*}
\end{defn}

We rewrite \eqref{1.1} as
\begin{eqnarray*}
u_t-u_{xxt}+u^2+2uu_x=uu_{xx}+u_{x}u_{xx}+uu_{xxx}.
\end{eqnarray*}
Denote $P(D)=(1-\partial_x^2)^{-1}$, \eqref{1.1} could be reformulated as
\begin{eqnarray} \label{1.3}
u_t+uu_x+F(u)=0,
\end{eqnarray}
where
\begin{eqnarray}\label{Fu}
F(u)=F_1(u)+F_2(u)+F_3(u)
\end{eqnarray}
with $F_1=\frac{1}{2}u^2$, $F_2=P(D)\left(\frac{1}{2}u^2+u_x^2\right)$ and
$F_3=\partial_x P(D)\left(\frac{1}{2}u^2+u_x^2\right)$.
\subsection{Local well-posedness in the Besov spaces}
Our main local existence result is given by the following theorem.
\begin{thm}\label{LSS}
Suppose that $1\leq p,r\leq\infty$ and $s>\max \left\{1+\frac{1}{p},\frac{3}{2}\right\}$. If
$u_0\in B_{p,r}^{s}$, then there exists a time $T>0$ such that the Cauchy problem \eqref{1.1} with initial data $u_0$ has a unique
solution $u\in E_{p,r}^s(T)$. The map $u_0\mapsto u$ is continuous from a neighborhood of $u_0$ in $B_{p,r}^{s}$ into
$$C([0, T];B_{p,r}^{s'})\cap C^1([0, T];B_{p,r}^{s'-1})$$
for every $s'<s$ when $r=+\infty$ and $s'=s$ whereas $r<+\infty$.
\end{thm}
\begin{re}
When $p=r=2$, the Besov space $B_{2,2}^{s}$ coincides with the Sobolev
space $H^s$. Theorem \ref{LSS} implies that
we can obtain the local well-posedness for the Cauchy problem \eqref{1.1} under the condition $u_0\in H^s$ with $s>\frac{3}{2}$.
\end{re}

In what follows, we denote $C>0$ as a generic constant only depending on $p,r,s$.
Uniqueness and continuity with respect to the initial data are an immediate corollary of the following
result.
\begin{pro}\label{pro1}
Suppose that $1\leq p,r\leq\infty$ and $s>\max \left\{1+\frac{1}{p},\frac{3}{2}\right\}$. Let
$u,v\in L^\infty([0,T);B_{p,r}^{s})\cap C([0,T];\mathscr{S}')$ be two given solutions to the Cauchy problem \eqref{1.1} with the initial data $u_0,v_0\in B_{p,r}^{s}$. Then for every $t\in [0,T]$, we have
\begin{eqnarray*}
\|u(t)-v(t)\|_{B_{p,r}^{s-1}}
\leq \|u_0-v_0\|_{B_{p,r}^{s-1}}exp\left\{C\int_{0}^{t}\|u(\tau)\|_{B_{p,r}^{s}}+\|v(\tau)\|_{B_{p,r}^{s}}d\tau\right\}.
\end{eqnarray*}
\end{pro}
\begin{proof}
Let $\omega=u-v$, then we have $\omega$ satisfies
\begin{eqnarray*}
\left\{
  \begin{array}{ll}
    \omega_t+v\partial_x\omega=-\omega\partial_xu-(F(u)-F(v)), \\
 \omega_0=u_0-v_0.
  \end{array}
\right.
\end{eqnarray*}
By using Lemma \ref{T}, we have
\begin{eqnarray}\label{3.1}
&&\|\omega(t)\|_{B_{p,r}^{s-1}}\nonumber\\
&\leq& e^{C\int_{0}^{t}\|\partial_xv(\tau)\|_{B_{p,r}^{s-2}}d\tau}
\times\left( \|\omega_0\|_{B_{p,r}^{s-1}}
+C\int_{0}^{t}e^{-C\int_{0}^{\tau}\|\partial_xv(s)\|_{B_{p,r}^{s-2}}ds}
\|\omega\partial_xu+(F(u)-F(v))\|_{B_{p,r}^{s-1}}d\tau\right).\nonumber\\
\end{eqnarray}
By using the imbedding  $B_{p,r}^{s-1}\hookrightarrow L^\infty$ with $s>1+\frac{1}{p}$, we have
\begin{eqnarray*}
\|\omega\partial_xu\|_{B_{p,r}^{s-1}}\leq
\|\omega\|_{B_{p,r}^{s-1}}\|u\|_{B_{p,r}^{s}}.
\end{eqnarray*}
Direct calculation we obtain
\begin{eqnarray*}
F(u)-F(v)=\frac{1}{2}\omega(u+v)+P(D)\left(\frac{1}{2}\omega(u+v)+\omega_x(u_x+v_x)\right)
+\partial_xP(D)\left(\frac{1}{2}\omega(u+v)+\omega_x(u_x+v_x)\right),
\end{eqnarray*}
then, for $s>\max \left\{1+\frac{1}{p},\frac{3}{2}\right\}$, we have
\begin{eqnarray*}
\|F(u)-F(v)\|_{B_{p,r}^{s-1}}\leq
C\|\omega\|_{B_{p,r}^{s-1}}(\|u\|_{B_{p,r}^{s}}+\|v\|_{B_{p,r}^{s}}),
\end{eqnarray*}
Hence,
\begin{eqnarray}\label{3.2}
\|\omega\partial_xu+(F(u)-F(v))\|_{B_{p,r}^{s-1}}\leq C\|\omega\|_{B_{p,r}^{s-1}}(\|u\|_{B_{p,r}^{s}}+\|v\|_{B_{p,r}^{s}}).
\end{eqnarray}
Substituting \eqref{3.2} into \eqref{3.1}, we find that
\begin{eqnarray*}
&& \exp\{{-C\int_{0}^{t}\|v(\tau)\|_{B_{p,r}^{s-1}}d\tau}\}\|\omega(t)\|_{B_{p,r}^{s-1}}\\
&\leq &\|\omega_0\|_{B_{p,r}^{s-1}}
+C\int_{0}^{t}\exp\{{-C\int_{0}^{\tau}\|v(\tau_1)\|_{B_{p,r}^{s-1}}d\tau_1}\}
\|\omega(\tau)\|_{B_{p,r}^{s-1}}(\|u(\tau)\|_{B_{p,r}^{s}}+\|v(\tau)\|_{B_{p,r}^{s}})d\tau.
\end{eqnarray*}
The proof of Proposition 2.1 is completed by using the Gronwall inequality.
\end{proof}

Motivated by the local existence theorem for Camassa-Holm equation in \cite{Dan2001}, we would like to use the classical Friedrichs
regularization method to construct the approximate solutions to the Cauchy problem \eqref{1.1}. The following lemma establishes the existence and compactness of approximate solutions which will be used in the proof of Theorem \ref{LSS}.
\begin{lem}\label{APS}
Let $u_0, p, r,s$ be the same as the statement of Theorem 2.1 and
$u^{(0)}= 0$. There exists a sequence of smooth functions $u^{(n)}\in C(\mathbb{R}^+;B_{p,r}^{\infty})$
solving
\begin{eqnarray}\label{3.3}
\left\{
  \begin{array}{ll}
  &u^{(n+1)}_t+u^{(n)}u^{(n+1)}_x=-F(u^{(n)}),\\
&u^{(n+1)}(0,x)=S_{n+1}u_0,
  \end{array}
\right.
\end{eqnarray}
where $F(u)$ is defined in \eqref{Fu} and $S_{n+1}$ is the low frequency cut-off operator defined in Proposition \ref{DB} . Moreover, there is a positive $T$ such that the solutions satisfying the following properties:
\begin{center}
(i) $(u^{(n)})_{n\in \mathbb{N}}$ is uniformly bounded in $E_{p,r}^s(T)$,~~~~~~~~~~\\
(ii) $(u^{(n)})_{n\in \mathbb{N}}$ is a Cauchy sequence in $C([0, T];B_{p,r}^{s-1})$.
\end{center}
\end{lem}


\begin{proof}
Since all the data $S_{n+1} u_0$ belong to $B_{p,r}^{\infty}$, from Lemma \ref{EU}, we obtain that for all $n\in\mathbb{N}$, equation \eqref{3.3} has a global solution $u_{n+1}$ belonging
to $C(\mathbb{R}^+;B_{p,r}^{\infty})$.
Making use of Lemma \ref{T}, we have
\begin{eqnarray}\label{3.4}
\|u^{(n+1)}(t)\|_{B_{p,r}^{s}}\leq e^{C\int_{0}^{t}\|u^{(n)}(\tau)\|_{B_{p,r}^{s}}d\tau}\times
\left( \|u_0^{(n+1)}\|_{B_{p,r}^{s}}+C\int_{0}^{t}e^{-C\int_{0}^{\tau}\|u^{(n)}(s)\|_{B_{p,r}^{s}}ds}
\|F(u^{(n)})\|_{B_{p,r}^{s}}d\tau\right).\nonumber\\
\end{eqnarray}
By Lemma \ref{Moser} and $B_{p,r}^{s-1}\hookrightarrow L^\infty$ with $s>1+\frac{1}{p}$, we estimate $\|F(u^{(n)})\|_{B_{p,r}^{s}}$ as
\begin{eqnarray}\label{3.5}
\|F(u^{(n)})\|_{B_{p,r}^{s}}&\leq& \|F_1(u^{(n)})\|_{B_{p,r}^{s}}+\|F_2(u^{(n)})\|_{B_{p,r}^{s}}+\|F_3(u^{(n)})\|_{B_{p,r}^{s}}\nonumber\\
&\leq& (\|u^{(n)}\|_{L^\infty}+\|u^{(n)}_x\|_{L^\infty})\|u^{(n)}\|_{B_{p,r}^{s}}\leq\|u^{(n)}\|_{B_{p,r}^{s}}^2.
\end{eqnarray}
Substituting \eqref{3.5} into \eqref{3.4}, we have
\begin{eqnarray*}
\|u^{(n+1)}(t)\|_{B_{p,r}^{s}}\leq e^{C\int_{0}^{t}\|u^{(n)}(\tau)\|_{B_{p,r}^{s}}d\tau}\times
\left( \|u_0\|_{B_{p,r}^{s}}+C\int_{0}^{t}e^{-C\int_{0}^{\tau}\|u^{(n)}(s)\|_{B_{p,r}^{s}}ds}
\|u^{(n)}\|_{B_{p,r}^{s}}^2d\tau\right).
\end{eqnarray*}
Let $T>0$ such that $2C \|u_0\|_{B_{p,r}^{s}}T<1$, we claim that for any $n\in \mathbb{N}$ and $t\in [0,T]$,
\begin{eqnarray}\label{3.7}
\|u^{(n)}\|_{B_{p,r}^{s}}\leq\frac{\|u_0\|_{B_{p,r}^{s}}}{1-2C\|u_0\|_{B_{p,r}^{s}}t}.
\end{eqnarray}
We use mathematical induction method to prove our claim. $u^{(0)}$ satisfies \eqref{3.7} naturally.
Assume \eqref{3.7} is true for $n$. We now prove that it also holds true for $n+1$.
\begin{eqnarray*}
\|u^{(n+1)}(t)\|_{B_{p,r}^{s}}&\leq& \frac{1}{\left(1-2C\|u_0\|_{B_{p,r}^{s}}t\right)^{1/2}}\times
\left( \|u_0\|_{B_{p,r}^{s}}+C\int_{0}^{t}\frac{\|u_0\|_{B_{p,r}^{s}}^2}
{\left(1-2C\|u_0\|_{B_{p,r}^{s}}\tau\right)^{3/2}}d\tau\right)\nonumber\\
&\leq& \frac{1}{\left(1-2C\|u_0\|_{B_{p,r}^{s}}t\right)^{1/2}}\times
\left( \|u_0\|_{B_{p,r}^{s}}+\int_{0}^{t}\frac{\|u_0\|_{B_{p,r}^{s}}}
{2\left(1-2C\|u_0\|_{B_{p,r}^{s}}\tau\right)^{3/2}}d\left(1-2C\|u_0\|_{B_{p,r}^{s}}\tau\right)\right)\nonumber\\
&\leq& \frac{ \|u_0\|_{B_{p,r}^{s}}}{1-2C\|u_0\|_{B_{p,r}^{s}}t}.
\end{eqnarray*}
Therefore, \eqref{3.7} is true and $(u^{(n)})_{n\in \mathbb{N}}$ is uniformly bounded in $C([0, T];B_{p,r}^{s})$.
It is easy to check that $u^{(n)}u^{(n)}_x$ is uniformly bounded in $C([0, T];B_{p,r}^{s-1})$ and $F(u^{(n)})$ is uniformly bounded in $C([0, T];B_{p,r}^{s-1})$. We conclude that the sequence $(u^{(n)})_{n\in \mathbb{N}}$ is uniformly bounded in
$C([0, T];B_{p,r}^{s})\cap C^1([0, T];B_{p,r}^{s-1})$.

Now, we turn to show $(u^{(n)})_{n\in \mathbb{N}}$ is a Cauchy sequence in $C([0, T];B_{p,r}^{s-1})$.
For $m,n\in \mathbb{N}$, from \eqref{3.3}, we have
\begin{eqnarray}\label{3.9}
\partial_t(u^{(n+m+1)}-u^{(n+1)})+u^{(n)}\partial_x(u^{(n+m+1)}-u^{(n+1)})
=-(u^{(n+m)}-u^{(n)})\partial_xu^{(n+m+1)}-(F(u^{n+m})-F(u^{n})).\nonumber\\
\end{eqnarray}
Making use of Lemma \ref{T}, we have
\begin{eqnarray*}
&&\|(u^{(n+m+1)}-u^{(n+1)})(t)\|_{B_{p,r}^{s-1}}\nonumber\\
&\leq& e^{C\int_{0}^{t}\|\partial_xu^{(n)}(\tau)\|_{B_{p,r}^{s-2}}d\tau}
\times\left( \|u_0^{(n+m+1)}-u_0^{(n+1)}\|_{B_{p,r}^{s-1}}
+C\int_{0}^{t}e^{-C\int_{0}^{\tau}\|\partial_xu^{(n)}(s)\|_{B_{p,r}^{s-2}}ds}
\|\tilde{F}\|_{B_{p,r}^{s-1}}d\tau\right),
\end{eqnarray*}
where $\tilde{F}=-(u^{(n+m)}-u^{(n)})\partial_xu^{(n+m+1)}-(F(u^{n+m})-F(u^{n}))$.

Then, we estimate $\|\tilde{F}\|_{B_{p,r}^{s-1}}$ as
\begin{eqnarray*}
\|\tilde{F}\|_{B_{p,r}^{s-1}}&\leq& \|(u^{(n)}-u^{(n+m)})\partial_xu^{(n+m+1)}\|_{B_{p,r}^{s-1}}+
\|F(u^{n+m})-F(u^{n})\|_{B_{p,r}^{s-1}}.
\end{eqnarray*}
By using the embedding $B_{p,r}^{s-1}\hookrightarrow L^\infty$ with $s>1+\frac{1}{p}$, we have
\begin{eqnarray*}
\|(u^{(n+m)}-u^{(n)})\partial_xu^{(n+m+1)}\|_{B_{p,r}^{s-1}}\leq
\|(u^{(n+m)}-u^{(n)})\|_{B_{p,r}^{s-1}}\|u^{(n+m+1)}\|_{B_{p,r}^{s}}.
\end{eqnarray*}
Direct calculation, we obtain
\begin{eqnarray*}
F(u^{n+m})-F(u^{n})&=&\frac{1}{2}(u^{(n+m)}-u^{(n)})(u^{(n+m)}+u^{(n)})\\
&&+P(D)\left(\frac{1}{2}(u^{(n+m)}-u^{(n)})(u^{(n+m)}+u^{(n)})+(u_x^{(n+m)}-u_x^{(n)})(u_x^{(n+m)}+u_x^{(n)})\right)\\
&&+\partial_xP(D)\left(\frac{1}{2}(u^{(n+m)}-u^{(n)})(u^{(n+m)}+u^{(n)})+(u_x^{(n+m)}-u_x^{(n)})(u_x^{(n+m)}+u_x^{(n)})\right).
\end{eqnarray*}
If $s>\max\{1+\frac{1}{p}, \frac{3}{2}\}$, then
\begin{eqnarray*}
\|F(u^{n+m})-F(u^{n})\|_{B_{p,r}^{s-1}}\leq
C\|u^{(n+m)}-u^{(n)}\|_{B_{p,r}^{s-1}}(\|u^{(n+m)}\|_{B_{p,r}^{s}}+\|u^{(n)}\|_{B_{p,r}^{s}}).
\end{eqnarray*}
Hence,
\begin{eqnarray*}
\|\tilde{F}\|_{B_{p,r}^{s-1}}\leq C\|u^{(n+m)}-u^{(n)}\|_{B_{p,r}^{s-1}}(\|u^{(n+m)}\|_{B_{p,r}^{s}}+\|u^{(n)}\|_{B_{p,r}^{s}}+\|u^{(n+m+1)}\|_{B_{p,r}^{s}}).
\end{eqnarray*}
We also have
\begin{eqnarray} \label{IVC}
\notag \|u_0^{(n+m+1)}-u_0^{(n+1)}\|_{B_{p,r}^{s-1}}&=& \|S_{n+m+1}u_0-S_{n+1}u_0\|_{B_{p,r}^{s-1}}
\\ \notag&=&\left\|\sum_{q=n+1}^{n+m}\Delta_q u_0\right\|_{B_{p,r}^{s-1}}
\\ \notag &=&\left(\sum_{k\geq -1}2^{rk(s-1)}\|\Delta_k \sum_{q=n+1}^{n+m}\Delta_q u_0\|_{L^p}^r \right)^{\frac{1}{r}}
\\ \notag&\leq& \left(\sum_{n\leq k\leq m+n+1}2^{rk(s-1)}(\|\Delta_k\Delta_{k-1}u_0\|_{L^p}+\|\Delta_k\Delta_{k+1} u_0\|_{L^p}+\|\Delta_{k}\Delta_k u_0\|_{L^p})^r \right)^{\frac{1}{r}}
\\ \notag&\leq& C\left(2^{-rn}\sum_{n\leq k\leq m+n+1}2^{rks}\|\Delta_k u_0\|_{L^p}^r \right)^{\frac{1}{r}}
\\&\leq& C2^{-n}\|u_0\|_{B_{p,r}^{s}}.
\end{eqnarray}
Since $u^{(n)}\in C([0, T];B_{p,r}^{s})$, 
consequently, for all $t\in[0,T]$,
\begin{eqnarray*}
\|(u^{(n+m+1)}-u^{(n+1)})(t)\|_{B_{p,r}^{s-1}}
\leq C\left( 2^{-n}
+\int_{0}^{t} \|(u^{(n+m)}-u^{(n)})\|_{B_{p,r}^{s-1}}d\tau\right).
\end{eqnarray*}
Arguing by induction with respect to the index $n$, one can easily prove that
\begin{eqnarray*}
\|(u^{(n+m+1)}-u^{(n+1)})(t)\|_{B_{p,r}^{s-1}}
\leq 2^{-n}\left(C\sum_{k=0}^{n}\frac{(2TC)^k}{k!}\right)+C\frac{(TC)^{n+1}}{(n+1)!}.
\end{eqnarray*}
Hence, $(u^{(n)})_{n\in \mathbb{N}}$ is a Cauchy sequence in $C([0, T];B_{p,r}^{s-1})$.
\end{proof}
{\bf\emph{ Proof of Theorem \ref{LSS}}}
From Lemma \ref{APS}, $(u^{(n)})_{n\in \mathbb{N}}$ is a Cauchy sequence in $C([0, T];B_{p,r}^{s-1})$. Then, $(u^{(n)})_{n\in \mathbb{N}}$ converges to a function $u\in C([0, T];B_{p,r}^{s-1})$.
Indeed, by using the locally compact embedding $B_{p,r}^{s}\hookrightarrow B_{p,r}^{s-1}$, the Arzela-Ascoli
theorem and Cantor’s diagonal process, we could extract a subsequence of $(u^{(n)})_{n\in \mathbb{N}}$ (we still use $(u^{(n)})_{n\in \mathbb{N}}$ to denote it.) and there exists $u \in Lip([0, T], B_{p,r}^{s-1})$, such that, for any $\phi \in C_0^{\infty}(\mathbb{R})$
\begin{eqnarray*}
\lim_{n\to \infty}	\|u^{(n)}(t)\phi-u(t)\phi\|_{B_{p,r}^{s-1}}=0,
\end{eqnarray*}
uniformly for $t \in [0, T]$.
 On the other hand, since $(u^{(n)})_{n\in \mathbb{N}}$ is uniformly bounded in $L^\infty([0, T];B_{p,r}^{s})$, it follows from the Fatou property for the Besov spaces(Lemma \ref{Fatou}) that $u\in L^\infty([0, T];B_{p,r}^{s})$.
 By interpolation, we could deduce for any $s'<s$
 \begin{eqnarray*}
\lim_{n\to \infty}\max_{t\in[0, T]}\|u^{(n)}(t)\phi-u(t)\phi\|_{B_{p,r}^{s'}}=0.
 \end{eqnarray*}
 Passing to the limit in \eqref{3.3} in the sense of
 distributions, we see that $u$ is a solution to \eqref{1.3} in the sense of
 distributions, one can easily get that $u_t\in L^\infty([0, T];B_{p,r}^{s-1})$. Lemma \ref{EU} guarantees that $u\in C([0, T];B_{p,r}^{s})$. From the equation \eqref{1.3}, one can also check $u_t\in C([0, T];B_{p,r}^{s-1})$. Thus, $u\in E_{p,r}^{s}(T)$. Uniqueness and continuity with respect to the initial data are an immediate consequence of Proposition \ref{pro1}.

\subsection{Local well-posedness in the critical Besov space}
%
This section will deal with the local existence result in the critical Besov space $B_{2,1}^{3/2}$, which is given in the following theorem.
\begin{thm}
Let $u_0\in B_{2,1}^{3/2}$, there exists a time $T>0$ such that the Cauchy problem (1.1) has a unique
solution $u\in E_{2,1}^{3/2}(T)$. The map $u_0\mapsto u$ is continuous from a neighborhood of $u_0$ in $B_{2,1}^{3/2}$ into
$C([0, T];B_{2,1}^{3/2})\cap C^1([0, T];B_{2,1}^{1/2}).$
\end{thm}
\begin{proof}
Let $(u^{(n)})_{n\in \mathbb{N}}$ be given by \eqref{3.3}. Since the proof is very long, we split it  into  four steps.

{\bf Step  1.}  By replacing $B_{p,r}^{s}$ with $B_{2,1}^{3/2}$ in Lemma \ref{APS}, we could prove that there exits a time $T>0$, such that the sequence $(u^{(n)})_{n\in \mathbb{N}}$ is uniformly bounded in
$C([0, T];B_{2,1}^{3/2})\cap C^1([0, T];B_{2,1}^{1/2})$.

{\bf Step 2.  In this step, we prove $(u^{(n)})_{n\in \mathbb{N}}$ is a Cauchy sequence in $C([0, T];B_{2,1}^{1/2})$.}

From \eqref{3.9} and Lemma \ref{T}, we have
\begin{eqnarray*}
&&\|(u^{(n+m+1)}-u^{(n+1)})(t)\|_{B_{2,\infty}^{1/2}}\nonumber\\&\leq& e^{C\int_{0}^{t}\|u^{(n)}(\tau)\|_{B_{2,\infty}^{3/2}}d\tau}
\times\left( \|u_0^{(n+m+1)}-u_0^{(n+1)}\|_{B_{2,\infty}^{1/2}}
+C\int_{0}^{t}e^{-C\int_{0}^{\tau}\|u^{(n)}(s)\|_{B_{2,\infty}^{3/2}}ds}
\|\tilde{F}\|_{B_{2,\infty}^{1/2}}d\tau\right),
\end{eqnarray*}
where $\tilde{F}=-(u^{(n+m)}-u^{(n)})\partial_xu^{(n+m+1)}-(F(u^{n+m})-F(u^{n}))$. Since $(u^{(n)})_{n\in \mathbb{N}}$ is uniformly bounded in
$C([0, T];B_{2,1}^{3/2})$, there exists an $M>0$ such that $\|u^{(n)}(t)\|_{B_{2,\infty}^{3/2}}\leq M$ for $t\in[0, T]$, which implies
\begin{eqnarray*}
e^{C\int_{0}^{t}\|u^{(n)}(\tau)\|_{B_{2,\infty}^{3/2}}d\tau}\leq e^{CMT}:=M_1.
\end{eqnarray*}
From Lemma \ref{Moser} and \ref{ine}, we obtain
\begin{eqnarray*}
\|\tilde{F}\|_{B_{2,\infty}^{1/2}}&\leq& \|(u^{(n+m)}-u^{(n)})\partial_xu^{(n+m+1)}\|_{B_{2,\infty}^{1/2}}+
\|(F(u^{n+m})-F(u^{n}))\|_{B_{2,\infty}^{1/2}}
\\&\leq& C \|(u^{(n+m)}-u^{(n)})\|_{B_{2,1}^{1/2}}(\|u^{(n+m+1)}\|_{B_{2,1}^{3/2}}+\|u^{(n)}\|_{B_{2,1}^{3/2}}+\|u^{(n+m)}\|_{B_{2,1}^{3/2}})\\
&\leq&  C\|(u^{(n+m)}-u^{(n)})\|_{B_{2,\infty}^{1/2}}\ln\left(e+\frac{\|(u^{(n+m)}-u^{(n)})\|_{B_{2,\infty}^{3/2}}}
{\|(u^{(n+m)}-u^{(n)})\|_{B_{2,\infty}^{1/2}}}\right).
\end{eqnarray*}
Consequently, for all $t\in[0,T]$,
\begin{eqnarray*}
&&\|(u^{(n+m+1)}-u^{(n+1)})(t)\|_{B_{2,\infty}^{1/2}}\\
&\leq& M_1 \left( \|u_0^{(n+m+1)}-u_0^{(n+1)}\|_{B_{2,\infty}^{1/2}}
+C\int_{0}^{t} {\|(u^{(n+m)}-u^{(n)})\|_{B_{2,\infty}^{1/2}}}\ln\left(e+\frac{\|(u^{(n+m)}-u^{(n)})\|_{B_{2,\infty}^{3/2}}}
{\|(u^{(n+m)}-u^{(n)})\|_{B_{2,\infty}^{1/2}}}\right)d\tau\right).
\end{eqnarray*}
By using similar calculation as \eqref{IVC} and Lemma \ref{EMT} ($ B_{2, 1}^{3/2} \hookrightarrow B_{2, \infty }^{3/2}$)
\begin{eqnarray*}
 \|u_0^{(n+m+1)}-u_0^{(n+1)}\|_{B_{2,\infty}^{1/2}}= \|S_{n+m+1}u_0-S_{n+1}u_0\|_{B_{2,\infty}^{1/2}}\leq C2^{-n}\|u_0\|_{B_{2,1}^{3/2}}.
\end{eqnarray*}
Setting $\omega_{n,m}(t)=\|(u^{(n+m)}-u^{(n)})(t)\|_{B_{2,\infty}^{1/2}}$, since $u^{(n)}\in C([0, T];B_{2,1}^{3/2})$ and the function $x\ln(e+\frac{2M}{x})$ is nondecreasing, then we have
\begin{eqnarray*}
\omega_{n+1,m}(t)\leq M_1\left( 2^{-n}
+C\int_{0}^{t} \omega_{n,m}(\tau)\ln\left(e+\frac{2M}
{\omega_{n,m}(\tau)}\right)d\tau\right).
\end{eqnarray*}
Setting $\omega_{n}(t)=\sup_{m\in\mathbb{N}}\omega_{n,m}(t)$ and $\tilde{\omega}(t)=\limsup_{n\rightarrow\infty}\omega_{n}(t)$.
Taking supremum with respect to $m$ in $\mathbb{N}$, $\omega_{n}(t)$ satisfies
\begin{eqnarray*}
\omega_{n+1}(t)\leq C M_1\left( 2^{-n}
+\int_{0}^{t} \omega_{n}(\tau)\ln\left(e+\frac{2M}
{\omega_{n}(\tau)}\right)d\tau\right).
\end{eqnarray*}
Which together with the Fatou–Lebesgue theorem implies
\begin{eqnarray*}
\tilde{\omega}(t)\leq CM_1\int_{0}^{t}\tilde{ \omega}(\tau)\ln\left(e+\frac{2M}
{\tilde{\omega}(\tau)}\right)d\tau.
\end{eqnarray*}
Making use of the Lemma \ref{Osgood} with $a=0$ and $\mu(r)=r\ln\left(e+\frac{2M}{r}\right)$, we have $\tilde{\omega}(t)=0$, for $\forall t\in[0,T]$,
that is to say  $(u^{(n)})_{n\in \mathbb{N}}$ is a Cauchy sequence in $C([0, T];B_{2,\infty}^{1/2})$. By Lemma 2.1, we see
\begin{eqnarray*}
\|(u^{(n+m)}-u^{(n)})\|_{B_{2,1}^{1/2}}&\leq& \|(u^{(n+m)}-u^{(n)})\|_{B_{2,1}^{1}}\\
&\leq& C\|(u^{(n+m)}-u^{(n)})\|_{B_{2,\infty}^{1/2}}^{1/2}\|(u^{(n+m)}-u^{(n)})\|_{B_{2,\infty}^{3/2}}^{1/2}\\
&\leq& C (2M)^{1/2}\|(u^{(n+m)}-u^{(n)})\|_{B_{2,\infty}^{1/2}}^{1/2}.
\end{eqnarray*}
Hence, $(u^{(n)})_{n\in \mathbb{N}}$ is a Cauchy sequence in $C([0, T];B_{2,1}^{1/2})$.

{\bf Step 3. In this step we prove the existence and uniqueness of the solution.}

We now prove the existence of solutions. From step 2, $(u^{(n)})_{n\in \mathbb{N}}$ is a Cauchy sequence in $C([0, T];B_{2,1}^{1/2})$. Then, $(u^{(n)})_{n\in \mathbb{N}}$ converges to a function $u\in C([0, T];B_{2,1}^{1/2})$. By using the same method that is used in the proof of Theorem \ref{LSS} we could obtain $u\in E_{2,1}^{3/2}(T)$.

For the uniqueness, assume that $u_1$ and $u_2$ are two solutions with the same initial datum $u_0$ to the system \eqref{1.1}. If we set $\varpi=u_1-u_2$, then we have $\varpi$ satisfies
\begin{eqnarray*}
\partial_t(\varpi)+u_2\partial_x\varpi
=-\varpi\partial_xu_1-(F(u_1)-F(u_2))
\end{eqnarray*}
with initial data $\varpi(x,0)=0$. Using the similar method that is used in the Step 2, we obtain
\begin{eqnarray*}
\|\varpi(t)\|_{B_{2,1}^{1/2}}
=0.
\end{eqnarray*}
This complete the proof of uniqueness.

{\bf Step 4. }
Continuity with respect to the initial data in $B_{2,1}^{3/2}$ can be obtained as in \cite{Dan2003}, we omit the details in this paper.

\end{proof}

\section{Ill-posedness in $B_{2,\infty}^{3/2}$} \label{IP}
\begin{thm}\label{IPB}
When $u_0\in{B_{2,\infty}^{3/2}}$, the Cauchy problem \eqref{1.1} is ill-posedness. More precisely, there
is a solution $u\in L^\infty([0,\infty];{B_{2,\infty}^{3/2}})$ of \eqref{1.1} such that for any $T$, $\varepsilon>0$, there exists a solution $v\in L^\infty([0,\infty];{B_{2,\infty}^{3/2}})$ with
\begin{eqnarray*}
\|v(0)-u(0)\|_{B_{2,\infty}^{3/2}}\leq \varepsilon \quad \emph{and} \quad \|v(t)-u(t)\|_{L([0,T];B_{2,\infty}^{3/2})}\geq 1.
\end{eqnarray*}
\end{thm}
\begin{re}
The proof of ill-posedness of CH equation is based on the solitary wave solution $u_c=ce^{-|x-ct|}$. However, such special solution is not a weak solution of \eqref{1.1}. Our proof relies on the single rogue peakon \eqref{roguepeakon} which is a solution in the sense of distribution. Moreover, the result of Theorem \ref{IPB} also indicates that $B_{2,1}^{3/2}$ is a critical space to consider the local well-posedness.
\end{re}
\begin{proof}
Let $A=\frac{1}{c}$ and $B=\ln|c|$ with $c\neq0$ in \eqref{roguepeakon}, we define $u_c(x,t):=\frac{c}{ct-1}e^{-|x-\ln|ct-1||}$, $c\neq0$, $t\geq0$.
Its Fourier transform in $x$ is
$$\hat{u}_c(x,t):=\frac{2c}{ct-1}\frac{e^{-i\xi \ln|ct-1|}}{1+\xi^2}.$$
When $p=2$, and $r\in[0,+\infty]$, the Definition \ref{besov} is equivalent to
\begin{eqnarray*}
\|u\|_{B_{2,r}^s}=\left\{
  \begin{array}{ll}
  \left[\left(\int_{-1}^1(1+\xi^2)^s|\hat{u}(\xi)|^2d\xi\right)^{\frac{r}{2}}
+\sum_{q\in\mathbb{N}}\left(\int_{2^q\leq|\xi|\leq2^{q+1}}(1+\xi^2)^s
|\hat{u}(\xi)|^2d\xi\right)^{\frac{r}{2}}\right]^{\frac{1}{r}}, & r\leq\infty, \\
   \max\left\{\left(\int_{-1}^1(1+\xi^2)^s|\hat{u}(\xi)|^2d\xi\right)^{\frac{1}{2}}, \, \sup_{q\in\mathbb{N}} \left(\int_{2^q\leq|\xi|\leq2^{q+1}}(1+\xi^2)^s
|\hat{u}(\xi)|^2d\xi\right)^{\frac{1}{2}}\right\}, & r=\infty.
  \end{array}
\right.
\end{eqnarray*}
For the initial data, we compute $\|u_{c_2}(0)-u_{c_1}(0)\|_{B_{2,\infty}^{3/2}}$ according to the above Definition.
\begin{eqnarray}\label{uc0}
\|u_{c_2}(0)-u_{c_1}(0)\|_{B_{2,\infty}^{3/2}}^2&=&8(c_2-c_1)^2\max\left\{\int_{0}^1\frac{1}{\sqrt{1+\xi^2}}d\xi
,\sup_{q\in\mathbb{N}} \int_{2^q}^{2^{q+1}}\frac{1}{\sqrt{1+\xi^2}}d\xi\right\}\nonumber\\
&=&8(c_2-c_1)^2\ln(1+\sqrt{2}).
\end{eqnarray}
Direct calculation, we have
\begin{eqnarray}\label{FT2}
\notag &&|\hat{u}_{c_2}(t)-\hat{u}_{c_1}(t)|^2=4\left|\frac{c_2}{c_2t-1}e^{-i\xi \ln|c_2t-1|}-\frac{c_1}{c_1t-1}e^{-i\xi \ln|c_1t-1|}\right|^2\\ \notag
&=&4\left(\frac{c_2}{c_2t-1}e^{-i\xi \ln|c_2t-1|}-\frac{c_1}{c_1t-1}e^{-i\xi \ln|c_1t-1|}\right)
\left(\frac{c_2}{c_2t-1}e^{i\xi \ln|c_2t-1|}-\frac{c_1}{c_1t-1}e^{i\xi \ln|c_1t-1|}\right)
\\&=&4\left(\frac{c_2}{c_2t-1}-\frac{c_1}{c_1t-1}\right)^2+\frac{8c_1c_2}{(c_2t-1)(c_1t-1)}
\left(1-\cos((\ln|c_2t-1|-\ln|c_1t-1|)\xi)\right),
\end{eqnarray}
For $T>0$, we choose $c_1=\frac{1}{1+T}$, $c_2=\frac{1+T-e^{2^{-q}\pi}}{(1+T)T}$ and $q$ large enough, then $c_1c_2>0$ and
\begin{eqnarray}\label{CL}
	\ln|c_2T-1|-\ln|c_1T-1|=2^{-q}\pi.
\end{eqnarray}
Substituting \eqref{FT2} and \eqref{CL} into above definition of $B_{2,\infty}^{3/2}$, we have
\begin{eqnarray*}
\|u_{c_2}(T)-u_{c_1}(T)\|_{B_{2,\infty}^{3/2}}&\geq& \frac{16c_1c_2}{(c_2T-1)(c_1T-1)}\int_{2^q}^{2^{q+1}}\frac{1-\cos(2^{-q}\pi\xi)}{\sqrt{1+|\xi|^2}}d\xi\\
&\geq& \frac{4c_1c_2}{\sqrt{2}(c_2T-1)(c_1T-1)},
\end{eqnarray*}
where we have used
\begin{eqnarray*}
&&\int_{2^q}^{2^{q+1}}\frac{\cos(2^{-q}\pi\xi)}{\sqrt{1+|\xi|^2}}d\xi\\
&=& \int_{\pi}^{2\pi}\frac{\cos t}{\sqrt{\pi^2+|2^q t|^2}}dt
\\
&=& \int_{\pi}^{\frac{3\pi}{2}}\frac{\cos t}{\sqrt{\pi^2+|2^q t|^2}}dt+ \int_{\frac{3\pi}{2}}^{2\pi}\frac{\cos t}{\sqrt{\pi^2+|2^q t|^2}}dt
\\
&=& -\int_{0}^{\frac{\pi}{2}}\frac{\cos \theta}{\sqrt{\pi^2+|2^q (\theta+\pi)|^2}}d\theta+ \int_{0}^{\frac{\pi}{2}}\frac{\cos \theta}{\sqrt{\pi^2+|2^q (2\pi-\theta)|^2}}d\theta
\\
&=& \int_{0}^{\frac{\pi}{2}}[\frac{1}{\sqrt{\pi^2+|2^q (2\pi-\theta)|^2}}-\frac{1}{\sqrt{\pi^2+|2^q (\theta+\pi)|^2}}]\cos \theta d\theta
\\&\leq& 0,
\end{eqnarray*}
and
\begin{eqnarray*}
\int_{2^q}^{2^{q+1}}\frac{1}{\sqrt{1+|\xi|^2}}d\xi \geq \frac{2^q}{\sqrt{1+2^{2(q+1)}}} \geq\frac{1}{4\sqrt{2}}.
\end{eqnarray*}
From \eqref{uc0}, we have
\begin{eqnarray*}
\|u_{c_2}(0)-u_{c_1}(0)\|_{B_{2,\infty}^{3/2}}&=&8\left(\frac{1+T-e^{2^{-q}\pi}}{(1+T)T}-\frac{1}{1+T}\right)^2\ln(1+\sqrt{2})\\
&=&8\left(\frac{1-e^{ 2^{-q}\pi}}{(1+T)T}\right)^2\ln (1+\sqrt{2}),
\end{eqnarray*}
which implies $\|u_{c_2}(0)-u_{c_1}(0)\|_{B_{2,\infty}^{3/2}}$ may be arbitrary small for $q$ large enough, however,
\begin{eqnarray*}
\|u_{c_2}(T)-u_{c_1}(T)\|_{B_{2,\infty}^{3/2}}\geq2.
\end{eqnarray*}

\end{proof}
\section{Global existence and blow-up phenomenon} \label{SGB}
Suppose $u(x,t)$ is a solution to the Cauchy problem \eqref{1.1} in its lifespan. Let us now consider the following differential equation:
\begin{eqnarray}\label{qt}
\left\{
  \begin{array}{ll}
   \frac{d q(x,t)}{dt}=u(q,t), \qquad 0<t<T, x\in \mathbb{R},\\
   q(x,0)=x, \qquad x\in \mathbb{R},
  \end{array}
\right.
\end{eqnarray}
where $T$ is the lifespan of the solution.
Taking derivative \eqref{qt} with respect to $x$, we obtain
$$\frac{dq_t}{dx}=q_{tx}=u_x(q,t)q_x, \qquad t\in(0,T).$$
Therefore
\begin{eqnarray*}\label{qx}
\left\{
  \begin{array}{ll}
   q_x=\exp\left(\int^t_0(u_x(q,s))ds\right), \qquad 0<t<T, x\in \mathbb{R},\\
   q_x(x,0)=1, \qquad x\in \mathbb{R},
  \end{array}
\right.
\end{eqnarray*}
which is always positive before the blow-up time. Therefore, the function $q(x,t)$ is an increasing diffeomorphism of the line before blow-up.
In fact, direct calculation yields
\begin{eqnarray*}
\frac{d}{dt}(e^{q(x,t)}m(q,t)q_x)
=e^{q(x,t)}[m_t(q,t)+m(q,t)u(q,t)+(m(q,t)u(q,t))_x]q_x=0.
\end{eqnarray*}
Hence, the following identity can be proved:
\begin{eqnarray}\label{IE}
e^{q(x,t)}m(q,t)q_x=e^x m_0(x),
\end{eqnarray}
which implies that $m(x, t)$ keeps the sign with respect to the initial datum.

By \eqref{IE}, we have
$$\int_{\mathbb{R}} e^{q(x,t)}m(q,t)q_xdx=\int_{\mathbb{R}} e^{q(x,t)}m(q,t)dq=\int_{\mathbb{R}}e^xm(x,t)dx.$$
It follows a conserved quantity as
$$\int_{\mathbb{R}}e^xm(x,t)dx=\int_{\mathbb{R}}e^xm_0dx.$$

We establish the blow-up scenario for the Cauchy problem \eqref{1.1} in the following Theorem.
\begin{thm}\label{blow-up scenario}
Let $u_0(x)\in H^s(\mathbb{R})$, $s>\frac{3}{2}$ and let $T$ be the maximal existence time of the solution $u(x,t)$ to \eqref{1.1} with the initial data
$u_0(x)$. Then the corresponding solution blows up in finite time if and only if
\begin{eqnarray*}
	\liminf _{t\rightarrow T-}\inf_{x\in\mathbb{R}}\left(u+\frac{1}{2}u_x\right)=-\infty.
\end{eqnarray*}
\end{thm}
\begin{proof}
	On one hand, suppose for any $t\in (0, T]$,  $\inf_{x\in\mathbb{R}}\left(u+\frac{1}{2}u_x\right) \geq -M$ with $M>0$.
Multiplying \eqref{1.1} by $m$, after integrating by parts, we have
\begin{eqnarray*}
\frac{1}{2}\frac{d}{dt}\|m\|_{L^2}^2&=&-\int_{\mathbb{R}}um^2+(um)_xm)dx\\
&=&-\int_{\mathbb{R}}\left(u+\frac{1}{2}u_x\right)m^2dx.
\end{eqnarray*}
By using Gronwall's inequality, we have
\begin{eqnarray*}
\|m(t)\|_{L^2}^2\leq\|m_0\|_{L^2}^2e^{2MT},
\end{eqnarray*}
then $\|u(t)\|_{H^2}^2\leq\|m(t)\|_{L^2}^2$ is bounded. This contradicts with the fact that $T$ is the maximal time of existence.

On the other hand, the solution does not blow up in $H^s, s>\frac{3}{2}$, that is $\|u\|_{H^s}$ is bounded, by Morrey's inequality, we have
\begin{eqnarray*}
\|u+\frac{1}{2}u_x\|_{L^\infty}\leq C\|u\|_{H^s}<+\infty.
\end{eqnarray*}
This completes the proof of Theorem \ref{blow-up scenario}.
\end{proof}

The operator $(1-\partial_x^2)^{-1} $ can be expressed by its associated Green's function as
$$u(x,t)=(1-\partial_x^2)^{-1} m(t,x)=G*m,\quad G=\frac{1}{2}e^{-|x|}.$$
More precisely,
\begin{eqnarray}\label{u}
u(x,t)=G\ast m(x,t)=\frac{1}{2}e^{-x}\int_{-\infty}^{x}e^{\xi}m(\xi,t)d\xi+\frac{1}{2}e^{x}\int_x^{\infty}e^{-\xi}m(\xi,t)d\xi.
\end{eqnarray}
\begin{eqnarray}\label{ux}
u_x(x,t)=-\frac{1}{2}e^{-x}\int_{-\infty}^{x}e^{\xi}m(\xi,t)d\xi+\frac{1}{2}e^{x}\int_x^{\infty}e^{-\xi}m(\xi,t)d\xi.
\end{eqnarray}
\begin{re}
	We know that for CH type equations, when $m_0(x)$ does not change sign, the solution usually exists globally. But for our system, we find a new a new phenomenon, which is if $m_0(x)\geq(\not\equiv) 0$, then the corresponding
	solution $u(x,t)$ exists globally, while, if $m_0(x)\leq(\not\equiv) 0$, then the corresponding
	solution $u(x,t)$ blows up in finite time.
\end{re}
\begin{thm}
Suppose that $u_0(x)\in H^s(\mathbb{R})$, $s>\frac{3}{2}$ , $m_0(x)\geq(\not\equiv) 0$. Then the corresponding
solution $u(x,t)$ to the Cauchy problem \eqref{1.1} with $u_0$ as the initial datum exists globally.
\end{thm}
\begin{proof}
By the blow-up scenario, we only need to show $u+\frac{1}{2}u_x$ has a lower bound.
From the initial data condition $m_0(x)\geq(\not\equiv) 0$ and \eqref{IE}, we have
\begin{eqnarray*}
m(x,t)\geq(\not\equiv) 0.
\end{eqnarray*}
By using \eqref{u} and \eqref{ux}, we have
\begin{eqnarray*}
u+\frac{1}{2}u_x=\frac{1}{4}e^{-x}\int_{-\infty}^{x}e^{\xi}m(\xi,t)d\xi+\frac{3}{4}e^{x}\int_x^{\infty}e^{-\xi}m(\xi,t)d\xi>0.
\end{eqnarray*}
\end{proof}

In the remainder of this section, we would like to establish a sufficient condition for the blow-up of solutions, which is given in the following Theorem.
\begin{thm}\label{BUT}
Suppose that $u_0(x)\in H^s(\mathbb{R})$, $s>\frac{3}{2}$ , $m_0(x)\leq(\not\equiv) 0$. Then, the corresponding solution $u(x, t)$ to the Cauchy problem \eqref{1.1} with $u_0$ as the initial datum blows up in finite time.
\end{thm}
\begin{proof}
From the condition $m_0(x)\leq(\not\equiv) 0$, we have, there exist a point $x_0\in\mathbb{R}$, such that
\begin{eqnarray}\label{111}
(u_0+\frac{1}{2}u_{0x})(x_0)=\frac{1}{4}e^{-x_0}\int_{-\infty}^{x_0}e^{\xi}m_0(\xi)d\xi+\frac{3}{4}e^{x_0}\int_{x_0}^{\infty}e^{-\xi}m_0(\xi)d\xi<0.
\end{eqnarray}
Rewrite \eqref{1.3} as
\begin{eqnarray}\label{ut}
u_t+uu_x+\frac{1}{2}u^2+G\ast \left(\frac{1}{2}u^2+u_x^2\right)+\partial_x G\ast \left(\frac{1}{2}u^2+u_x^2\right) =0.
\end{eqnarray}
Differentiating \eqref{ut} with respect to $x$, we obtain
\begin{eqnarray*}
u_{xt}+uu_{xx}+u_x^2+uu_x+\partial_x G\ast \left(\frac{1}{2}u^2+u_x^2\right)+\partial_x^2 G\ast \left(\frac{1}{2}u^2+u_x^2\right) =0,
\end{eqnarray*}
it follows
\begin{eqnarray*}
u_{xt}+uu_{xx}-\frac{1}{2}u^2+uu_x+ G\ast \left(\frac{1}{2}u^2+u_x^2\right)+\partial_x G\ast \left(\frac{1}{2}u^2+u_x^2\right) =0.
\end{eqnarray*}
Then, we have $\frac{d\left(u+\frac{1}{2}u_x\right)}{dt}$ at the point $q(x_0,t)$.
\begin{eqnarray}\label{5.10}
&&\frac{d\left(u+\frac{1}{2}u_x\right)(q(x_0,t),t)}{dt}\nonumber\\
&=&u_t(q(x_0,t),t)+uu_x(q(x_0,t),t)+\frac{1}{2}u_{xt}(q(x_0,t),t)
+\frac{1}{2}uu_{xx}(q(x_0,t),t)\nonumber\\
&=&-\frac{1}{4}u^2(q(x_0,t),t)-\frac{1}{2}uu_x(q(x_0,t),t)- \frac{3}{2}G\ast \left(\frac{1}{2}u^2+u_x^2\right)(q(x_0,t),t)-\frac{3}{2}\partial_x G\ast \left(\frac{1}{2}u^2+u_x^2\right)(q(x_0,t),t).\nonumber\\
\end{eqnarray}
Note that,
\begin{eqnarray}\label{5.11}
&&G\ast\left(\frac{1}{2}u^2+u_x^2\right)(q(x_0,t),t)\nonumber\\
&=&\frac{1}{2}e^{-q(x_0,t)}\int_{-\infty}^{q(x_0,t)}e^{\xi}\left(\frac{1}{2}u^2+u_\xi^2\right)d\xi
+\frac{1}{2}e^{q(x_0,t)}\int_{q(x_0,t)}^{\infty}e^{-\xi}\left(\frac{1}{2}u^2+u_\xi^2\right)d\xi
\end{eqnarray}
and
\begin{eqnarray}\label{5.12}
&&\partial_x G\ast\left(\frac{1}{2}u^2+u_x^2\right)(q(x_0,t),t)\nonumber\\
&=&-\frac{1}{2}e^{-q(x_0,t)}\int_{-\infty}^{q(x_0,t)}e^{\xi}\left(\frac{1}{2}u^2+u_\xi^2\right)d\xi
+\frac{1}{2}e^{q(x_0,t)}\int_{q(x_0,t)}^{\infty}e^{-\xi}\left(\frac{1}{2}u^2+u_\xi^2\right)d\xi.
\end{eqnarray}
Substituting \eqref{5.11} and \eqref{5.12} into \eqref{5.10}, we obtain
\begin{eqnarray}\label{5.13}
\frac{d\left(u+\frac{1}{2}u_x\right)(q(x_0,t),t)}{dt}
&=&-\frac{1}{4}u^2(q(x_0,t),t)
-\frac{1}{2}uu_x(q(x_0,t),t)-\frac{3}{2}e^{q(x_0,t)}\int_{q(x_0,t)}^{\infty}e^{-\xi}\left(\frac{1}{2}u^2+u_\xi^2\right)d\xi\nonumber\\
&\leq& -\frac{5}{8}u^2(q(x_0,t),t)
-\frac{1}{2}uu_x(q(x_0,t),t)-\frac{9}{8}e^{q(x_0,t)}\int_{q(x_0,t)}^{\infty}e^{-\xi}u_\xi^2d\xi\nonumber\\
&\leq& -\frac{5}{8}u^2(q(x_0,t),t)-\frac{1}{2}uu_x(q(x_0,t),t),
\end{eqnarray}
where, we have used
\begin{eqnarray*}
\int_{x}^{\infty}e^{-\xi}\left(u^2+u_x^2\right)d\xi\geq \int_{x}^{\infty}e^{-\xi}\left(-2uu_\xi\right)d\xi
\geq e^{-x}u^2(x,t)-\int_{x}^{\infty}e^{-\xi}u^2d\xi
\end{eqnarray*}
yields
\begin{eqnarray*}
 e^{x}\int_{x}^{\infty}e^{-\xi}\left(2u^2+u_{\xi}^2\right)d\xi
\geq u^2(x,t).
\end{eqnarray*}
Under the initial data condition $m_0(x,t)\leq(\not\equiv) 0$ and \eqref{IE}, we have
\begin{eqnarray*}
m(x,t)\leq(\not\equiv) 0,
\end{eqnarray*}
which implies that for any $x\in\mathbb{R}$,
 \begin{eqnarray}\label{5.14}
(u^2-u_x^2)(x,t)=\int_{-\infty}^{x}e^{\xi}m(\xi,t)d\xi\int_x^{\infty}e^{-\xi}m(\xi,t)d\xi>0.
\end{eqnarray}
Then, from \eqref{5.13} and \eqref{5.14}, we have
\begin{eqnarray*}
\frac{d\left(u+\frac{1}{2}u_x\right)(q(x_0,t),t)}{dt}&\leq& -\frac{1}{2}\left(u+\frac{1}{2}u_x\right)^2(q(x_0,t),t)-\frac{1}{8}u^2(q(x_0,t),t)+\frac{1}{8}u_x^2(q(x_0,t),t)\nonumber\\
&\leq &-\frac{1}{2}\left(u+\frac{1}{2}u_x\right)^2(q(x_0,t),t).
\end{eqnarray*}
By \eqref{111} and using standard arguments for ordinary differential equations, it is easy to conclude that there exists a finite time $T$, such that
\begin{eqnarray*}
\lim_{t\rightarrow T}\left(u+\frac{1}{2}u_x\right)(q(x_0,t),t)=-\infty \quad with \quad T\leq \frac{-2}{(u_0+\frac{1}{2}u_{0x})(x_0)}.
\end{eqnarray*}
This completes the proof of Theorem \ref{BUT}.
\end{proof}

\section{Conclusion and open problem}
In this paper, we develop a new type of soliton solution -- called {\bf rogue peakon} in the integrable peakon theory.  The  rogue peakon takes on a rational form at the peak of
a logarithmic function, instead of a regular traveling wave, and is found in the integrable Camassa-Holm (CH) type equation (\ref{1.0}). We also provide multi-rogue peakon solutions with some interesting interactional dynamics as shown in Figures 2 - 6. Moreover, we discuss the local well-posedness of the solution in the Besov space $B_{p,r}^{s}$ with $1\leq p,r\leq\infty$, $s>\max \left\{1+1/p,3/2\right\}$ or $B_{2,1}^{3/2}$, and then prove the ill-posedness of the solution in $B_{2,\infty}^{3/2}$. It is worth to point out that the proof of ill-posedness of relies on the single rogue peakon \eqref{roguepeakon} which is a solution in the sense of distribution. For the classical Camassa-Holm equation, Degasperis-Procesi equation and Holm-Staley $b$-family equation, when $m_0(x)$ does not change sign, the solution usually exists globally. But for our system, we find a new phenomenon, which is if $m_0(x)\geq(\not\equiv) 0$, then the corresponding
	solution $u(x,t)$ exists globally, while if $m_0(x)\leq(\not\equiv) 0$, then the corresponding
	solution $u(x,t)$ blows up in a finite time. But the global rogue peakon stability keeps open.

\section{Acknowledgments}
This work was partially supported by the National Natural Science Foundation of China (Grant No.12071439 and Grant No. 11971475). The authors really very much appreciate the fruitful discussions with Prof. Senyue Lou (NBU), Prof. Hongyu Liu (CUHK), Prof. Jingsong He (SZU), Prof. Ruguang Zhou (JSNU), Prof. Zuonong Zhu (SHJT), Prof. Yufeng Zhang (CUMT), and Prof. Jifeng Chu (SHNU).

$ $

\section*{Appendix}
\appendix
\renewcommand{\thelem}{A.\arabic{lem}}
\renewcommand{\thepro}{A.\arabic{pro}}
\renewcommand{\thedefn}{A.\arabic{defn}}
For consistence of our paper, we recall some basic results on the Littlewood-Paley decomposition theory and
Besov spaces which required in the proof of our main theorems.
\begin{pro}[Littlewood–Paley decomposition]\emph{\cite{Dan2001,Dan03}} \label{DB}
	Let $\mathcal{B}:=\{\xi\in \mathbb{R}^d;|\xi|\leq\frac{4}{3}\}$ and $\mathcal{C}:=\{\xi\in \mathbb{R}^d;\frac{8}{3}\leq|\xi|\leq\frac{8}{3}\}$. Then there exist two smooth radial functions $\chi\in C_c^\infty(\mathcal{B})$ and $\varphi\in C_c^\infty(\mathcal{C})$ valued in $[0, 1]$, such that
	\begin{eqnarray*}
		&&\chi(\xi)+\sum_{q\geq0}\varphi(2^{-q}\xi)=1,\quad \forall \xi\in \mathbb{R}^d,\\
		&& \mathrm{supp} \varphi(2^{-q}\cdot)\cap \mathrm{supp } \varphi(2^{-q'}\cdot)=\varnothing,\,\, if \,\, |q-q'|\geq 2,\\
		&&\mathrm{supp}\chi(\cdot)\cap\mathrm{supp} \varphi(2^{-q}\cdot)=\varnothing,\,\, if \,\, q\geq1,\\
		&&\frac{1}{2}\leq \chi^2(\xi)+\sum_{q\geq0}\varphi^2(2^{-q}\xi)\leq 1,\quad \forall \xi \in \mathbb{R}^d.
	\end{eqnarray*}
	We denote the Fourier transform and its inverse by $\mathcal{F}$ and $ \mathcal{F}^{-1}$. Let $h:=\mathcal{F}^{-1}\varphi$ and $\tilde{h}:=\mathcal{F}^{-1}\chi$. Then for all $f\in \mathscr{S}'(\mathbb{R}^d)$, the dyadic operators $\Delta_q$ and $S_q$ can be defined as follows
	\begin{eqnarray*}
		&&\Delta_qf:=\varphi(2^{-q}D)f=2^{qd}\int_{\mathbb{R}^d}h(2^qy)f(x-y)dy\quad if \quad q\geq0,\\
		&&S_qf:=\chi(2^{-q}D)f=\sum_{-1\leq k\leq q-1}\Delta_kf=2^{qd}\int_{\mathbb{R}^d}\tilde{h}(2^qy)f(x-y)dy,\\
		&&\Delta_{-1}f:=S_0f \quad and \quad \Delta_qf:=0\quad if \quad q\leq-2.
	\end{eqnarray*}
	Hence,
	\begin{eqnarray*}
		f=\sum_{q\in \mathbb{Z}}\Delta_qf
	\end{eqnarray*}
	holds in $\mathscr{S}'(\mathbb{R}^d)$ and is called the Littlewood-Paley decomposition.
\end{pro}

\begin{defn}[Besov spaces]\emph{\cite{Dan2001,Dan03}}\label{besov}
	Let $1\leq p,r\leq\infty$ and $s\in\mathbb{R}$. The nonhomogeneous Besov space
	$B_{p,r}^{s}(\mathbb{R}^d)$ is defined by
	\begin{eqnarray*}
		B_{p,r}^{s}(\mathbb{R}^d):=\{f\in\mathscr{S}'(\mathbb{R}^d);\|f\|_{B_{p,r}^{s}}<\infty \},
	\end{eqnarray*}
	where
	\begin{eqnarray*}
		\|f\|_{B_{p,r}^{s}}
		:=\left\{
		\begin{array}{ll}
			\left(\sum_{q\in \mathbb{Z}}2^{qsr}\|\Delta_q f \|_{L^p}^r\right)^{\frac{1}{r}}, & for \quad r<\infty, \\
			\sup_{q\in \mathbb{Z}}2^{qs}\|\Delta_q f \|_{L^p}, & for \quad r=\infty.
		\end{array}
		\right.
	\end{eqnarray*}
	If $s=\infty$, $B_{p,r}^{\infty}:=\cap_{s\in\mathbb{R}}B_{p,r}^{s}$.
\end{defn}
Then, we present some useful properties of $B_{p,r}^{s}$.
\begin{lem}\emph{\cite{Dan2001,Dan03,Tri}}\label{EMT}
	Let $s\in\mathbb{R}$ and $1\leq p,r,p_j,r_j\leq\infty$, $j=1,2$. Then the following properties
	hold:\\
	\emph{(1)} $B_{p,r}^{s}(\mathbb{R})$ is a Banach space and is continuously embedded in $\mathscr{S}'(\mathbb{R})$.\\
	\emph{(2)} $C_c^\infty(\mathbb{R})$ is dense in $B_{p,r}^{s}(\mathbb{R})\Leftrightarrow 1\leq p,r<\infty$.\\
	\emph{(3)} $B_{p_1,r_1}^{s_1}(\mathbb{R})\hookrightarrow B_{p_2,r_2}^{s_2}(\mathbb{R})$, if $p_1\leq p_2$, $r_1\leq r_2$ and $s_2=s_1-\left(\frac{1}{p_1}-\frac{1}{p_2}\right)$.
	
	 $B_{p,r_1}^{s_1}(\mathbb{R})\hookrightarrow B_{p,r_2}^{s_2}(\mathbb{R})$, if $s_2<s_1$ or $s_2=s_1$, $r_1\leq r_2$.
	
	$B_{p_1,r_1}^{s_1}(\mathbb{R})\hookrightarrow B_{p_2,r_2}^{s_2}(\mathbb{R})$ is locally compact if $s_2\leq s_1$; $r_1\leq r_2$.\\
	\emph{(4)} $\forall s>0$, $B_{p,r}^{s}(\mathbb{R})\cap L^\infty(\mathbb{R})$ is a Banach algebra. Moreover, $B_{p,r}^{s}(\mathbb{R})$ is a Banach algebra$\Leftrightarrow$ $B_{p,r}^{s}(\mathbb{R})\hookrightarrow L^{\infty}(\mathbb{R})$ $\Leftrightarrow$ $s>\frac{1}{p}$(or $s\geq\frac{1}{p}$ and $r=1$).\\
	\emph{(5)} $\forall$ $\theta \in [0,1]$, $s=\theta s_1+(1-\theta)s_2$,
	$$\|f\|_{B_{p,r}^{s}}\leq C \|f\|_{B_{p,r}^{s_1}}^\theta \|f\|_{B_{p,r}^{s_2}}^{1-\theta},\quad \forall f\in B_{p,r}^{s_1}\cap B_{p,r}^{s_2}.$$\\
	\emph{(6)} $\forall$ $\theta \in (0,1)$, $s_1>s_2$, $s=\theta s_1+(1-\theta)s_2$,
	$$\|f\|_{B_{p,1}^{s}}\leq \frac{C(\theta)}{s_1-s_2}\|f\|_{B_{p,\infty}^{s_1}}^\theta \|f\|_{B_{p,\infty}^{s_2}}^{1-\theta},\quad \forall f\in B_{p,\infty}^{s_1}.$$
\end{lem}
\begin{lem}\emph{\cite{BCD2011,Dan2001,Dan2003}}\label{T}
	Let $1\leq p,r\leq\infty$ and $s>-\min \left\{\frac{1}{p},1-\frac{1}{p}\right\}$. Assume that
	$f_0\in B_{p,r}^{s}$, $F\in L^1([0,T];B_{p,r}^{s})$ and $\partial_x v\in L^1([0,T];B_{p,r}^{s-1})$ if $s>1+\frac{1}{p}$ or
	$\partial_x v\in L^1([0,T];B_{p,r}^{1/p}\cap L^\infty)$ otherwise. If $f\in L^\infty([0,T];B_{p,r}^{s})\cap C([0,T]; \mathscr{S}')$ solves the following 1D linear transport equation:
	\begin{eqnarray}\label{LTE}
		\left\{
		\begin{array}{ll}
			f_t+vf_x=F,t>0,x\in\mathbb{R},\\
			f(0,x)=f_0, x\in\mathbb{R},
		\end{array}
		\right.
	\end{eqnarray}
	then there exists a constant C depending only on $s,p,r$ such that the following statements hold:\\
	\emph{(1)} For all $t\in[0,T]$,
	\begin{eqnarray*}
		\|f\|_{B_{p,r}^{s}}\leq \|f_0\|_{B_{p,r}^{s}}+\int_{0}^t \|F(\tau)\|_{B_{p,r}^{s}}d\tau+C\int_{0}^t V'(\tau) \|f(\tau)\|_{B_{p,r}^{s}}d\tau,
	\end{eqnarray*}
	and hence,
	\begin{eqnarray*}
		\|f\|_{B_{p,r}^{s}}\leq e^{CV(t)}\left( \|f_0\|_{B_{p,r}^{s}}+\int_{0}^t e^{-C V(\tau)} \|F(\tau)\|_{B_{p,r}^{s}}d\tau\right),
	\end{eqnarray*}
	where
	\begin{eqnarray*}
		V(t)=\left\{
		\begin{array}{ll}
			\int_{0}^t  \|v_x(\tau)\|_{B_{p,r}^{1/p}\cap L^\infty}d\tau,\quad if ~~ s<1+\frac{1}{p}, \\
			\int_{0}^t  \|v_x(\tau)\|_{B_{p,r}^{s-1}}d\tau,\quad if ~~ s>1+\frac{1}{p}(or~~ s=1+\frac{1}{p} ~~and~~ r=1).
		\end{array}
		\right.
	\end{eqnarray*}
	\emph{(2)} If $f=v$, then (1) holds true for all $s > 0$ with $V(t)=\int_{0}^t \|v_x(\tau)\|_{L^\infty}d\tau$.
\end{lem}

\begin{lem}\label{EU}\emph{\cite{BCD2011,Dan03}}
	Let $p, r, s, f_0$ and $F$ be as in the statement of Lemma \ref{T}. Assume that $v\in L^\rho([0,T];B_{\infty,\infty}^{-M})$ for some $\rho>1$, $M>0$ and $v_x\in L^1([0,T];B_{p,r}^{s-1})$ if $s>1+\frac{1}{p}$ or $s=1+\frac{1}{p}$, $r=1$ and $v_x\in L^1([0,T];B_{p,\infty}^{1/p}\cap L^\infty)$ if
	$s<1+\frac{1}{p}$. Then, equation \eqref{LTE} have a unique solution
	$f\in L^\infty([0,T];B_{p,r}^s)\cap C([0,T];B_{p,1}^{s'})$ for any $s'<s$ and the inequalities of Lemma \ref{T} hold true. Moreover, if $r<\infty$, then $f\in C([0,T];B_{p,r}^{s})$.
\end{lem}
\begin{lem}[Moser-type estimates]\emph{\cite{Che2004,Dan2003,Dan03}}\label{Moser}
	Let $1\leq p,r\leq\infty$, then we have the following estimates:\\
	\emph{ (1)} For $s>0$, $\|fg\|_{B_{p,r}^{s}}\leq C(\|f\|_{B_{p,r}^{s}}\|g\|_{L^\infty}+\|f\|_{L^\infty}\|g\|_{B_{p,r}^{s}})$, $\forall f,g \in B_{p,r}^{s}\cap L^\infty$.\\
	\emph{(2)} For all $s_1\leq\frac{1}{p}< s_2$ ($s_2\geq\frac{1}{p}$ if $r=1$) and $s_1+s_2>0$,
	$$\|fg\|_{B_{p,r}^{s_1}}\leq C\|f\|_{B_{p,r}^{s_1}}\|g\|_{B_{p,r}^{s_2}}, \forall f\in B_{p,r}^{s_1}, g\in B_{p,r}^{s_2}.$$
\end{lem}

\begin{lem}\emph{\cite{Dan2003,Dan03}}\label{ine}
	There is a constant $C > 0$ such that for $s\in\mathbb{R}$, $\varepsilon>0$ and $1\leq p\leq \infty$,
	$$\|f\|_{B_{p,1}^{s}}\leq C\frac{1+\varepsilon}{\varepsilon}
	\|f\|_{B_{p,\infty}^{s}}\ln\left(e+\frac{\|f\|_{B_{p,\infty}^{s+\varepsilon}}}{\|f\|_{{B_{p,\infty}^{s}}}}\right),\quad f\in B_{p,\infty}^{s+\varepsilon}.$$
\end{lem}

\begin{lem}[Osgood Lemma]\emph{\cite{BCD2011,Che1998}}\label{Osgood}
	Let $\rho$ be a measurable function from $[t_0,T]$ to $[0, c]$, $\gamma$ be a locally integrable
	function from $[t_0,T]$ to $\mathbb{R}^+$, $\mu$ be a continuous and nondecreasing function from $[0, c]$ to $\mathbb{R}^+$, and $a\geq 0$ be a real number. Assume that $\rho$ satisfies
	\begin{eqnarray*}
		\rho(t)\leq a +\int_{t_0}^t \gamma(t')\mu(\rho(t'))dt', \quad a.e. \quad t\in [t_0,T].
	\end{eqnarray*}
	If $a>0$, then for a.e. $t\in [t_0,T]$ we have
	\begin{eqnarray*}
		-\mathcal{M}(\rho(t))+\mathcal{M}\mathcal(a)\leq \int_{t_0}^{t}\gamma(t')dt',\quad with \quad \mathcal{M}(x)=\int_x^c\frac{dr}{\mu(r)}.
	\end{eqnarray*}
	If $a=0$ and $\mu$ satisfies the condition $\int_0^1\frac{dr}{\mu(r)}=+\infty$, then $\rho=0$ a.e. $t\in [t_0,T]$.

\end{lem}

\begin{lem}[Fatou-type lemma]\emph{\cite{Dan2003,Dan03}}\label{Fatou}
	If $\{u^{(k)}\}$ is a bounded sequence in $B_{p,r}^s$, which tends to $u$ in $\mathscr{S}'$, then $u\in B_{p,r}^s$ and
	\begin{eqnarray*}
		\|u\|_{B_{p,r}^s}\leq \liminf_{k\rightarrow\infty} \|u^{(k)}\|_{B_{p,r}^s}.
	\end{eqnarray*}
\end{lem}

\end{document}